\documentclass[a4paper,USenglish]{lipics-v2018}
\nolinenumbers
\usepackage{amsfonts}
\usepackage{amsthm}
\usepackage{amsmath}
\usepackage{algorithm}
\usepackage[noend]{algpseudocode}
\usepackage{csquotes}
\usepackage[disable]{todonotes}

\newcommand\katz{\mathbf{c}}
\newcommand\eg{e.g.\ }
\newcommand\ie{i.e.\ }
\newcommand\etal{et al.\ }
\newcommand\etals{et al.'s\ }

\EventEditors{Yossi Azar, Hannah Bast, and Grzegorz Herman}
\EventNoEds{3}
\EventLongTitle{26th Annual European Symposium on Algorithms (ESA 2018)}
\EventShortTitle{ESA 2018}
\EventAcronym{ESA}
\EventYear{2018}
\EventDate{August 20--22, 2018}
\EventLocation{Helsinki, Finland}
\EventLogo{}
\SeriesVolume{112}
\ArticleNo{42} 

\author{Alexander van der Grinten}%
{Department of Computer Science, Humboldt-Universität zu Berlin, Germany}%
{avdgrinten@hu-berlin.de}{}{}

\author{Elisabetta Bergamini}%
{Karlsruhe Institute of Technology (KIT), Germany}%
{}{}{}

\author{Oded Green}%
{School of Computational Science and Engineering, 
Georgia Institute of Technology, USA}%
{ogreen@gatech.edu}{}{}

\author{David A. Bader}%
{School of Computational Science and Engineering, Georgia Institute of Technology, USA}%
{bader@gatech.edu}{}{}

\author{Henning Meyerhenke}%
{Department of Computer Science, Humboldt-Universität zu Berlin, Germany}%
{meyerhenke@hu-berlin.de}{}{}

\authorrunning{A. van der Grinten, E. Bergamini, O. Green, D. A. Bader, H. Meyerhenke}

\title{Scalable Katz Ranking Computation in Large Static and Dynamic Graphs}
\titlerunning{Scalable Katz Ranking Computation in Large Static and Dynamic Graphs}

\Copyright{Alexander van der Grinten, Elisabetta Bergamini, Oded Green, David A. Bader and Henning Meyerhenke}

\subjclass{,
}


\funding{Most of the work was done while AvdG was affiliated with University of Cologne,
	Germany, and HM with Karlsruhe Institute of Technology and University of Cologne.
	Additionally, EB, AvdG and HM were partially supported by grant ME 3619/3-2 within
	German Research Foundation (DFG) Priority Programme 1736. Funding was also provided by
	Karlsruhe House of Young Scientists via the
	International Collaboration Package. 
	Funding for OG and DB was provided in part by the Defense Advanced Research Projects Agency (DARPA)
	under Contract Number FA8750-17-C-0086. The content of the information in this document does not necessarily
	reflect the position or the policy of the Government, and no official endorsement should be inferred.
	The U.S. Government is authorized to reproduce and distribute reprints for Government purposes notwithstanding
	any copyright notation here on.}

\begin{document}

\maketitle

\begin{abstract}
	Network analysis defines a number of centrality measures
	to identify the most central nodes in a network.
	Fast computation of those measures is a major
	challenge in algorithmic network analysis.
	Aside from closeness and betweenness,
	Katz centrality is one of the established centrality measures.
	In this paper, we consider the problem of computing
	rankings for Katz centrality.
	In particular, we propose upper and lower bounds on
	the Katz score of a given node. While previous approaches
	relied on numerical approximation or heuristics
	to compute Katz centrality rankings, we construct an algorithm that iteratively
	improves those upper and lower bounds until a correct Katz
	ranking is obtained.
	We extend our algorithm to dynamic graphs while maintaining its
	correctness guarantees.
	Experiments demonstrate that our static graph algorithm outperforms
	both numerical approaches and heuristics with speedups between
	$1.5\times$ and $3.5\times$, depending on the desired quality guarantees.
	Our dynamic graph algorithm improves upon the static algorithm
	for update batches of less than 10000 edges.
	We provide
	efficient parallel CPU and GPU implementations of our algorithms
	that enable near real-time Katz centrality computation for graphs
	with hundreds of millions of nodes in fractions of seconds.
\end{abstract}


\section{Introduction}
	Finding the most important nodes of a network is a major task in network analysis. To this end,
	numerous centrality measures have been introduced in the literature.  Examples of well-known
	measures are betweenness (which ranks nodes according to their participation in the shortest
	paths of the network) and closeness (which indicates the average shortest-path distance to other
	nodes).  A major limitation of both measures is that they are based on the assumption that
	information flows through the networks following shortest paths only. However, this is often not
	the case in practice; think, for example, of traffic on street networks: it is easy to imagine
	reasons why drivers might prefer to take slightly longer paths.  On the other hand, it is also
	quite unlikely that \textit{much} longer paths will be taken.
	 
	Katz centrality~\cite{katz1953defn} accounts for this
	by summing all walks starting from a node, but weighting them based on their length. More
	precisely, the weight of a walk of length $i$ is $\alpha^i$, where $\alpha$ is some attenuation
	factor smaller than $1$. Thus, naming $\omega_i(v)$ the number of walks of length $i$ starting
	from node $v$, the Katz centrality of $v$ is defined as
	\begin{equation}
		\label{eq:katz}
		\katz(v) := \sum_{i=1}^\infty \omega_i(v)\,\alpha^i
	\end{equation}
	or equivalently: $\katz = \left(\sum_{i=1}^\infty A^i\,\alpha^i\right)\vec{I}$, where $A$ is the
	adjacency matrix of the graph and $\vec{I}$ is the vector consisting only of $1$s.
	This can be restated as a Neumann series, resulting in
	the closed-form expression $\katz = \alpha A (I - \alpha A)^{-1} \vec{I}$,
	where $I$ is the identity
	matrix.
	Thus, Katz centrality can be computed exactly by solving the linear system
	\begin{equation}
		\label{eq:la}
		(I - \alpha A)\,\mathbf{z} = \vec{I} ~,
	\end{equation}
	followed by evaluating $\katz = \alpha A\,\mathbf{z}$.
	We call this approach the \emph{linear algebra formulation}.
	In practice, the solution to Eq.~(\ref{eq:la}) is numerically approximated
	using iterative solvers for linear systems.
	While these solvers yield solutions of good quality,
	they can take hundreds of iterations to converge \cite{nathan2017guarantee}.
	Thus, in terms of running time, those algorithms can be impractical for
	today's large networks, which often have millions of nodes and billions of edges.

	Instead, Foster \etals~\cite{foster2001iter} algorithm
	estimates Katz centrality
	iteratively by computing partial sums of the series from Eq.~(\ref{eq:katz})
	until a stopping criterion is reached.
	Although very
	efficient in practice, this method has no guarantee on the correctness of the ranking it finds,
	not even for the top nodes.
	Thus, the approach is ineffective for applications where only
	a subset of the most central nodes is needed or when accuracy is needed.
	As this is indeed the case in many applications, several
	top-$k$ centrality algorithms have been proposed recently for
	closeness~\cite{bergamini2016topclose}
	and betweenness~\cite{lee2014topbetween}.
	Recently, a top-$k$ algorithm for Katz centrality~\cite{nathan2017guarantee}
	was suggested. That algorithm
	still relies on solving Eq.~(\ref{eq:la}); however, it reduces
	the numerical accuracy that is required to obtain a top-$k$ rating.
	Similarly, Zhan \etal~\cite{zhan2017topkatz} propose a heuristic method to exclude
	certain nodes from top-$k$ rankings but do not present
	algorithmic improvements on the actual Katz computation.

	\subparagraph*{Dynamic graphs.}
	Furthermore, many of today's real-world networks,
	such as social networks and web graphs, are dynamic in nature and some
	of them evolve over time at a very quick pace. For such networks, it is often impractical to
	recompute centralities from scratch after each graph modification. Thus, several dynamic graph
	algorithms that efficiently update centrality have been introduced for
	closeness~\cite{bisenius2018dynclose} and
	betweenness~\cite{lee2016dynbetween}.
	Such algorithms usually work well in practice, because they
	reduce the computation to the part of the graph that has actually been affected.
	This offers potentially large speedups compared to recomputation.
	For Katz centrality, dynamic algorithms
	have recently been proposed by Nathan \etal~\cite{nathan2017dynsolve,nathan2018dynperson}.
	However, those algorithms rely on heuristics and are unable to
	reproduce the exact Katz ranking after dynamic updates.

	\subparagraph*{Our contribution.}
		We construct a vertex-centric
		algorithm that computes Katz centrality by iteratively
		improving upper and lower bounds on the centrality scores (see Section~\ref{sec:static}
		for the construction of this algorithm).
		While the computed centralities are approximate, our algorithm guarantees
		the correct ranking. We extend (in Section \ref{sec:dyn})
		this algorithm to dynamic graphs while preserving the guarantees of the static algorithm.
		An extensive experimental evaluation (see Section~\ref{sec:exp})
		shows that (i) our new algorithm outperforms Katz algorithms
		that rely on numerical approximation with speedups between $1.5\times$ and $3.5\times$,
		depending on the desired correctness guarantees,
		(ii) our algorithm has a speedup in the same order of magnitude
		over the widely-used heuristic
		of Foster \etal~\cite{foster2001iter} while improving accuracy,
		(iii) our dynamic graph algorithm improves upon static recomputation
		of Katz rankings for batch sizes of less than 10000 edges
		and (iv) efficient parallel CPU and GPU implementations of our algorithm
		allow near real-time computation of Katz centrality in fractions
		of seconds even for very large graphs. In particular, our GPU implementation
		achieves speedups of more than $10\times$ compared to a 20-core CPU implementation.


\section{Preliminaries}
\label{sec:prelim}

\subsection{Notation}
	\begin{description}
		\item[Graphs] In the following sections, we assume that
			$G = (V, E)$ is the input graph to our algorithm.
			Unless stated otherwise, we
			assume that $G$ is directed.
			For the purposes of Katz centrality, undirected graphs can be modeled
			by replacing each undirected edge with two directed edges in reverse directions.
			For a node $x \in V$, we denote the \emph{out-}degree of $x$ by
			$\deg(x)$. The maximum out-degree of any node in $G$ is denoted by $\deg_\mathrm{max}$.

		\item[Katz centrality] The Katz centrality of the nodes of $G$
			is given by Eq.~(\ref{eq:katz}). With $\katz_i(v)$ we denote
			the $i$-th partial sum of Eq.~(\ref{eq:katz}).
			Katz centrality is not defined for
			arbitrary values of $\alpha$. In general, Eq.~(\ref{eq:katz}) converges
			for $\alpha < \frac1{\sigma_\mathrm{max}}$, where
			$\sigma_\mathrm{max}$ is the largest singular value of
			the adjacency matrix $A$ (see~\cite{katz1953defn}).
			
			Katz centrality can also be defined by counting \emph{inbound}
			walks in $G$~\cite{katz1953defn,newman2010networks}.
			For this definition, $\omega_i(x)$ is replaced by
			the number of walks of length $i$ that end in $x \in V$.
			Indeed, for applications like web graphs, nodes that are the target
			of many links intuitively should be considered more central than nodes that
			only have many links
			themselves\footnote{This is a central idea behind the
				PageRank~\cite{brin1998google} metric.}.
			However, as inbound Katz centrality coincides with the outbound Katz centrality
			of the reverse graph, we will not specifically consider it in this paper.
			\todo{Maybe remove inbound Katz.}
	\end{description}

\subsection{Related work}
	Most algorithms that are able to compute Katz scores with
	approximation guarantees are based on the linear algebra formulation and
	compute a numerical solution to Eq.~(\ref{eq:la}).
	Several approximation algorithms have been developed in order to
	decrease the practical running times of this formulation
	(\eg based on low-rank approximation~\cite{acer2009lowrank}).
	Nathan \etal~\cite{nathan2017guarantee} prove a relationship between
	the numerical approximation quality of Eq.~(\ref{eq:la})
	and the resulting Katz ranking quality.
	While this allows computation of top-$k$ rankings with reduced numerical
	approximation quality, no significant speedups can be expected
	if full Katz rankings are desired.

	Foster \etal~\cite{foster2001iter} present a vertex-centric heuristic
	for Katz centrality:
	They propose to determine Katz centrality by computing the recurrence
	$c_{i+1} = \alpha A\ c_i + \vec{I}$.
	The computation is iterated until either a fixed
	point\footnote{Note that a true fixed point will not be reached using this method
		unless the graph is a DAG.}
	or a predefined number of iterations is reached.
	This algorithm performs well in
	practice; however, due to the heuristic nature of the stopping condition,
	the algorithm does not give any correctness guarantees.

	Another paper from Nathan \etal~\cite{nathan2018dynperson}
	discusses an algorithm for a \enquote{personalized} variant of Katz centrality.
	Our algorithm uses a similar iteration scheme but differs in
	multiple key properties of the algorithm: Instead of considering
	personalized Katz centrality, our algorithm computes the usual,
	\enquote{global} Katz centrality. While Nathan \etal give a global
	bound on the quality of their solution, we are able to compute
	per-node bounds that can guarantee the correctness of our ranking.
	Finally, Nathan \etals dynamic update procedure is a heuristic
	algorithm without correctness guarantee, although its ranking quality
	is good in practice. In contrast to that,
	our dynamic algorithm reproduces exactly the results of the
	static algorithm.
	\todo{Maybe shorten this paragraph.}


\section{Iterative improvement of Katz bounds}
	\label{sec:static}

\subsection{Per-node bounds for Katz centrality}
	The idea behind our algorithm is to compute upper and lower bounds on the centrality
	of each node. Those bounds are iteratively improved.
	We stop the iteration once an application-specific stopping criterion is reached.
	When that happens,
	we say that the algorithm \emph{converges}.
	
	Per-node upper and lower bounds allow us to rank nodes against each other:
	Let $\ell_r(x)$ and $u_r(x)$ denote respectively lower and upper bounds
	on the Katz score of node $x$ after iteration $r$.
	An explicit construction of those bounds will be given later in this section;
	for now, assume that such bounds exist.
	Furthermore, let $w$ and $v$ be two nodes; without loss of generality,
	we assume that $w$ and $v$ are chosen such that $\ell_r(w) \geq \ell_r(v)$.
	If $\ell_r(w) > u_r(v)$, then $w$ appears in the Katz centrality
	ranking before $v$ and we say that $w$ and $v$ are \emph{separated} by the bounds
	$\ell_r$ and $u_r$.
	In this context, it should be noted that per-node bounds do not allow us to prove that
	the Katz scores of two nodes are
	equal\footnote{In theory, the linear algebra formulation is able to prove that
		the score of two nodes is indeed equal.
		However, in practice, limited floating point precision
		limits the usefulness of this property.}.
	However, as the algorithm still needs to be able to rank nodes $x$ that share
	the same $\ell_r(x)$ and $u_r(x)$ values,
	we need a more relaxed concept of separation. Therefore:
	\begin{definition}
		In the same setting as before, let $\epsilon > 0$.
		We say that $w$ and $v$ are $\epsilon$-\emph{separated},
		if and only if
		\begin{equation}
			\label{eq:epsilon_sep}
			\ell_r(w) > u_r(v) - \epsilon ~.
		\end{equation}
	\end{definition}
	Intuitively, the introduction
	of $\epsilon$ makes the $\epsilon$-condition easier to fulfill than the
	separation condition: Indeed, separated pairs of nodes are
	also $\epsilon$-separated for every $\epsilon > 0$.
	In particular, $\epsilon$-separation allows us to construct Katz rankings even in the presence
	of nodes that have the same Katz score: Those nodes are
	never separated, but they will eventually be $\epsilon$-separated for every $\epsilon > 0$.

	In order to actually construct rankings, it is sufficient to notice that
	once all pairs of nodes are $\epsilon$-separated, sorting the nodes
	by their lower bounds $\ell_r$ yields a correct Katz ranking,
	except for pairs of nodes with a difference in Katz score of less than $\epsilon$.
	Thus, using this definition, we can discuss possible stopping criteria for the algorithm:

	\begin{description}
		\item[Ranking criterion.] Stop once all nodes are $\epsilon$-separated from each other.
			This guarantees that the ranking
			is correct, except for nodes with scores that are very close to each other.
		\item[Top-$k$ criterion.] Stop once the top-$k$ nodes are $\epsilon$-separated from
			each other and from all other nodes. For $k = n$ this criterion reduces to the
			ranking criterion.
		\item[Score criterion.] Stop once the difference between the upper and lower bound
			of each node becomes less than $\epsilon$. This guarantees that the
			Katz centrality of each node is correct up to an additive constant of $\epsilon$.
		\item[Pair criterion.] Stop once two given nodes $u$ and $v$ are
			$\epsilon$-separated.
	\end{description}

	First, we notice that a simple lower bound on the
	Katz centrality of a node $v$ can be obtained by truncating the series in Eq.~(\ref{eq:katz}) after
	$r$ iterations, hence, $\ell_r(v) := \sum_{i=1}^{r} \omega_i(v) \alpha^i$ is a
	lower bound on $\katz(v)$.
	For undirected graphs, this lower bound can be improved to
	$\sum_{i=1}^r \omega_i(v) \alpha^i + \omega_r(v) \alpha^{r+1}$,
	as any walk of length $r$ can be extended to a walk of length $r+1$ with the
	same starting point by repeating its last edge with reversed
	direction.

	\begin{theorem}
		\label{thm:correctness}
		Let $\gamma = \frac{\deg_\mathrm{max}}{1 - \alpha \deg_\mathrm{max}}$.
		For any $r \geq 1$, $v \in V$ and $\alpha < \frac{1}{\deg_\mathrm{max}}$,
		the value
		\[ u_r(v) := \sum_{i=1}^{r} \alpha^i \omega_i(v) + \alpha^{r+1}\omega_r(v) \gamma \]
		is an upper bound on $\katz(v)$.
	\end{theorem}
	\begin{proof}
		First, let $S_i(v)$ be the set of nodes $x$ for which there exists a walk of length $i$
		starting in $v$ and ending in $x$.
		Each walk of length $i+1$ is the concatenation of a walk of
		length $i$ ending in $x \in S_i(v)$
		and an edge $(x, y)$, where $y$ is some neighbor of $x$.
		Let $\omega_i(v, x)$ denote the number of walks of length $i$
		that start in $v$ and end in $x$. Thus, we can write 
		\begin{equation}
			\label{eq:deg_ineq}
			\omega_{i+1}(v) = \sum_{x \in S_i(v)} \deg(x)\ \omega_i(v, x)
				\leq \sum_{x \in S_i(v)} \deg_\mathrm{max}\ \omega_i(v, x)
				= \deg_\mathrm{max}\ \omega_i(v) ~.
		\end{equation}
		By applying induction to the previous inequality,
		it is easy to see that, for any $j > 1$,
		\[ \omega_{i+j}(v) \leq (\deg_\mathrm{max})^j \omega_i(v) ~.\]
		Discarding the first $r$ terms of the sum in Eq.~(\ref{eq:katz})
		then yields
		\begin{align*}
			\sum_{i=r+1}^{\infty} \alpha^i \omega_i(v)
			& \leq \sum_{j=1}^{\infty} \alpha^{r+j} (\deg_\mathrm{max})^j \omega_r(v)
			= \alpha^r \omega_r(v) \sum_{j=1}^{\infty} (\alpha \deg_\mathrm{max})^j\\
			& = \alpha^r \omega_r(v) \left( \frac{1}{1- \alpha \deg_\mathrm{max}} - 1\right)
			= \alpha^{r+1} \omega_r(v) \gamma ~.
		\end{align*}
		For the second to last equality, we rewrite the infinite series as a geometric sum.
	\end{proof}

	The following lemma (with proof in Appendix~\ref{sec:proofs})
	shows that we can indeed iteratively improve
	the upper and lower bounds for each node $x \in V$:

	\begin{lemma}
		\label{lem:monotone}
		For each $x \in V$, $\ell_i(x)$ is non-decreasing in $i$ and
		$u_i(x)$ is non-increasing in $i$.
	\end{lemma}
	
	Theorem~\ref{thm:correctness} requires us to choose $\alpha < \frac1{\deg_\mathrm{max}}$,
	which is a restriction compared to the more general
	requirement of $\alpha < \frac1{\sigma_\mathrm{max}}$.
	For our experiments in the later sections of this paper,
	we set $\alpha = \frac1{1+\deg_\mathrm{max}}$ in order to satisfy this condition.
	Aside from enabling us to apply the theorem, this choice of $\alpha$ has
	some additional advantages: First, because Theorem~\ref{thm:correctness} gives
	an upper bound on Eq.~(\ref{eq:katz}), Katz centrality is
	guaranteed to converge for this value of the $\alpha$
	parameter\footnote{This was already noticed by Katz~\cite{katz1953defn}
		and can alternatively be proven through linear algebra.}.
	$\deg_\mathrm{max}$ is also much easier to compute than
	$\sigma_\mathrm{max}$, an operation that is comparable in complexity
	to computing the Katz centrality
	itself\footnote{Indeed, the popular power iteration method to compute
		$\sigma_\mathrm{max}$ for real, symmetric, positive-definite matrices
		has a complexity of $\Omega(r\,|E|)$, where $r$ denotes a number
		of iterations.}.
	Finally, $\alpha = \frac1{1+\deg_\mathrm{max}}$ is widely-used in existing
	literature \cite{bonchi2012personal,foster2001iter}, with Foster \etal calling it the
	\enquote{generally-accepted default attenuation factor}~\cite{foster2001iter}.

	It is worth remarking (proof in Appendix~\ref{sec:proofs})
	that graphs exist for which the bound from Theorem~\ref{thm:correctness}
	is sharp:
	\begin{lemma}
		If $G$ is a complete graph, $u_i(x) = \katz(x)$ for all $x \in V$ and $i \in \mathbb{N}$.
	\end{lemma}

\subsection{Efficient rankings using per-node bounds}
	In the following, we state the description of our Katz algorithm for static graphs.
	As hinted earlier, the algorithm estimates Katz centrality by computing
	$u_r(v)$ and $\ell_r(v)$. These upper and lower bounds are
	iteratively improved by incrementing $r$ until the algorithm converges.
	
	To actually compute $\katz_r(v)$, we use the well-known fact that the number of walks of
	length $i$ starting in node $v$ is equal to the sum of the number of walks of length $i-1$
	starting in the neighbors of $v$, in other words:
	\begin{equation}
		\label{eq:basic_recurrence}
		\omega_i(v) = \sum_{v \to x \in E} \omega_{i-1}(x)~.
	\end{equation}
	Thus, if we initialize $\omega_1(v)$ to $\deg(v)$ for all $v \in V$, we can then
	repeatedly loop over the edges of $G$ and compute tighter and tighter lower bounds.

	We focus here on the top-$k$ convergence criterion.
	It is not hard to see how our
	techniques can be adopted to the other stopping criteria
	mentioned at the start of the previous subsection.
	To be able to efficiently detect convergence,
	the algorithm maintains a set of \emph{active} nodes.
	These are the nodes for which the lower and upper bounds have not yet converged.
	Initially, all nodes are active.
	Each node is \emph{deactivated} once it is $\epsilon$-separated from the $k$ nodes
	with highest lower bounds $\ell_r$.
	It should be noted that, because of Lemma \ref{lem:monotone},
	deactivated nodes will stay deactivated in all future iterations.
	Thus, for the top-$k$ criterion, it is sufficient
	to check whether (i) only $k$ nodes remain active and (ii) the remaining
	active nodes are $\epsilon$-separated from each other.
	This means that each iteration will require less work than its previous iteration.

	\begin{algorithm}[t] 
		\caption{Katz centrality bound computation for static graphs}
		\label{algo:topk}
		\begin{minipage}{0.475\textwidth}
			\begin{algorithmic}
				\footnotesize
				\State $\gamma \gets \deg_\mathrm{max}/(1 - \alpha \deg_\mathrm{max})$
				\smallskip
				\State Initialize $\katz_0(x) \gets 0 \quad \forall x \in V$
				\State Initialize $r \gets 0$ and $\omega_0(x) \gets 1 \quad \forall x \in V$
				\State Initialize set of active nodes: $M \gets V$
				\While{not \Call{converged}{\null}}
					\State Set $r \gets r+1$ and $\omega_{r}(x) \gets 0 \quad \forall x \in V$
					\ForAll{$v \in V$}
						\ForAll{$v \to u \in E$}
							\State $\omega_{r}(v) \gets \omega_{r}(v) + \omega_{r-1}(u)$
						\EndFor
						\State $\katz_r(v) \gets \katz_{r-1}(v) + \alpha^r \omega_r(v)$
						\If{$G$ undirected}
							\State $\ell_r(v) \gets \katz_r(v) + \alpha^{r+1} \omega_r(v)$
						\Else
							\State $\ell_r(v) \gets \katz_r(v)$
						\EndIf
						\State $u_r(v) \gets \katz_r(v)  + \alpha^{r+1} \omega_r(v) \gamma$
					\EndFor
				\EndWhile
			\end{algorithmic}
		\end{minipage}\hspace{0.05\textwidth}%
		\begin{minipage}{0.475\textwidth}
			\begin{algorithmic}
				\footnotesize
				\Function{converged}{\null}
					\State \Call{partialSort}{$M$, $k$, $\ell_r$, decreasing}
					\ForAll{$i \in \{k+1, \ldots, |V|\}$}
						\If{$u_r(M[i]) - \epsilon < \ell_r(M[k])$}
							\State $M \gets M \setminus \{v\}$
						\EndIf
					\EndFor
					\If{$|M| > k$}
						\State \Return $\mathrm{false}$
					\EndIf
					\ForAll{$i \in \{2, \ldots, \min(|M|, k)\}$}
						\If{$u_r(M[i]) - \epsilon \geq \ell_r(M[i - 1])$}
							\State \Return $\mathrm{false}$
						\EndIf
					\EndFor
					\State \Return $\mathrm{true}$
				\EndFunction
			\end{algorithmic}
		\end{minipage}
	\end{algorithm}

	Algorithm \ref{algo:topk} depicts the pseudocode of the algorithm.
	Computation of $\omega_r(v)$ is done by evaluating the
	recurrence from Eq.~(\ref{eq:basic_recurrence}).
	After the algorithm terminates,
	the $\epsilon$-separation property guarantees that
	the $k$ nodes with highest $\ell_r(v)$
	form a top-$k$ Katz centrality ranking (although $\ell_r(v)$
	does not necessarily equal the true Katz score).

	The \Call{converged}{} procedure in Algorithm~\ref{algo:topk}
	checks whether the top-$k$ convergence criterion is satisfied.
	In this procedure, $M$ denotes the set of active nodes.
	The procedure first partially sorts the elements of $M$ by decreasing
	lower bound $\ell_r$.
	After that is done, the first $k$ elements of $M$ correspond
	to the top-$k$ elements in the current ranking
	(which might not be correct yet).
	Note that it is not necessary to construct the entire
	ranking here; sorting just the top-$k$ nodes is sufficient.
	The procedure tries to deactivate nodes that cannot be in the top-$k$
	and afterwards checks if the remaining top-$k$ nodes
	are correctly ordered. These checks are performed by testing if the
	$\epsilon$-separation condition from Eq.~(\ref{eq:epsilon_sep}) is true.

	\subparagraph*{Complexity analysis.}
	The sequential worst-case time complexity of Algorithm~\ref{algo:topk}
	is $\mathcal{O}(r\,|E| + r\,\mathcal{C})$,
	where $r$ is the number of iterations and $\mathcal{C}$ is the complexity
	of the convergence checking procedure.
	It is easy to see that the loop over
	$V$ can be parallelized, yielding a complexity of
	$\mathcal{O}(r\,\frac{|V|}p \deg_{max} + r\,\mathcal{C})$
	on a parallel machine with $p$ processors.
	The complexity of \Call{converged}{}, the top-$k$ ranking convergence criterion,
	is dominated by the
	$\mathcal{O}(|V| + k \log k)$ complexity of partial sorting.
	Both the score and the pair criteria can be implemented in $\mathcal{O}(1)$.

	It should be noted that -- for the same solution quality --
	our algorithm converges at least as fast as
	the heuristic of Foster \etal that computes a Katz
	ranking without correctness guarantee.
	Indeed, the values of $\katz_r$ yield exactly the values
	that are computed by the heuristic. However, Foster \etals
	heuristic is unable to accurately assess the quality of
	its current solution and might thus perform too many or too
	few iterations.


\section{Updating Katz centrality in dynamic graphs}
	\label{sec:dyn}

	In this section, we discuss how our Katz centrality algorithm can be extended
	to compute Katz centrality rankings for dynamically
	changing graphs. We model those graphs as an initial graph that is
	modified by a sequence of edge insertions and edge deletions.
	We do not explicitly handle node insertions and deletions as those can
	easily be supported by adding enough isolated nodes to the initial graph.

	Before processing any edge updates, we assume that our algorithm
	from Section~\ref{sec:static} was first executed on the initial
	graph to initialize the values $\omega_i(x)$ for all $x \in V$.
	The dynamic graph algorithm needs to recompute $\omega_i(x)$
	for $i \in \{1, \ldots, r\}$, where $r$ is the number of iterations that
	was reached by the static Katz algorithm on the initial graph.
	The main observation here is that if an edge $u \to v$ is inserted into (or deleted from)
	the initial graph, $\omega_i(x)$ only changes for nodes $x$ in the vicinity
	of $u$. More precisely, $\omega_i(x)$ can only change if $u$ is reachable
	from $x$ in at most $i-1$ steps.

	Algorithm~\ref{algo:dynamic} depicts the pseudocode of our dynamic Katz algorithm.
	$\mathcal{I}$ denotes the set of edges to be inserted, while $\mathcal{D}$ denotes
	the set of edges to be deleted. We assume that 
	$\mathcal{I} \cap E = \emptyset$ and $\mathcal{D} \subseteq E$ before the algorithm.
	Effectively, the algorithm performs a breadth-first search (BFS) through the
	reverse graph of $G$ and updates $\omega_i$ for all nodes nodes that
	were reached in steps $1$ to $i$.

	\begin{algorithm}[t]
		\caption{Dynamic Katz update procedure}
		\label{algo:dynamic}
		\begin{minipage}{0.475\textwidth}
			\begin{algorithmic}
				\footnotesize
				\State $E \gets E \setminus \mathcal{D}$
				\State $S \gets \emptyset$,\ $T \gets \emptyset$
				\ForAll{$w \to v \in \mathcal{I} \cup \mathcal{D}$}
					\State $S \gets S \cup \{w\}$
					\State $T \gets T \cup \{v\}$
				\EndFor
				\ForAll{$i \in \{1, \ldots, r\}$}
					\State \Call{updateLevel}{i}
				\EndFor
				\ForAll{$w \in S$}
					\State Recompute $\ell_r(w)$ and $u_r(w)$ from $\katz_r(w)$
				\EndFor
				\ForAll{$w \in V$}
					\If{$u_r(w) \geq \min_{x \in M} \ell_r(x) - \epsilon$}
						\State $M \gets M \cup \{w\}$ \Comment{Reactivation}
					\EndIf
				\EndFor
				\State $E \gets E \cup \mathcal{I}$
				\While{not \Call{converged}{\null}}
					\State Run more iterations of static algorithm
				\EndWhile
			\end{algorithmic}
		\end{minipage}\hspace{0.05\textwidth}%
		\begin{minipage}{0.475\textwidth}
			\begin{algorithmic}
				\footnotesize
				\Procedure{updateLevel}{i}
					\ForAll{$v \in S \cup T$}
						\State $\omega_i'(v) \gets \omega_i(v)$
					\EndFor
					\ForAll{$v \in S$}
						\ForAll{$w \to v \in E$}
							\State $S \gets S \cup \{w\}$
							\State $\omega'_i(w) \gets \omega'_i(w)
									- \omega_{i-1}(v) + \omega'_{i-1}(v)$
						\EndFor
					\EndFor
					\ForAll{$w \to v \in \mathcal{I}$}
						\State $\omega'_i(w) \gets \omega'_i(w) + \omega'_{i-1}(v)$
					\EndFor
					\ForAll{$w \to v \in \mathcal{D}$}
						\State $\omega'_i(w) \gets \omega'_i(w) - \omega_{i-1}(v)$
					\EndFor
					\ForAll{$w \in S$}
						\State $\katz_i(w) \gets \katz_i(w)
								- \alpha^i \omega_i(w) + \alpha^i \omega'_i(w)$
					\EndFor
				\EndProcedure
			\end{algorithmic}
		\end{minipage}
	\end{algorithm}

	After the update procedure terminates, the new upper and lower bounds
	can be computed from $\katz_r$ as in the static algorithm.
	We note that $\omega'_i(x)$ matches exactly the value
	of $\omega_i(x)$ that the static Katz algorithm would compute for the modified graph.
	Hence, the dynamic algorithm reproduces the correct values
	of $\katz_r(x)$ and also of $\ell_r(x)$ and $u_r(x)$ for all $x \in V$.
	In case of the top-$k$ convergence criterion,
	some nodes might need to be \emph{reactivated} afterwards: Remember
	that the top-$k$ criterion maintains a set $M$ of active nodes.
	After edge updates are processed, it can happen that there are nodes $x$ that
	are not $\epsilon$-separated from all nodes in $M$ anymore.
	Such nodes $x$ need to be added to $M$ in order to obtain a
	correct ranking. The ranking itself can then be updated by sorting
	$M$ according to decreasing $\ell_r$.
	
	It should be noted that there is another related corner case:
	Depending on the convergence criterion, it can happen that the algorithm
	is not converged anymore even after nodes have been reactivated.
	For example, for the top-$k$ criterion, this is the case if the nodes
	in $M$ are not $\epsilon$-separated from each other anymore.
	Thus, after the dynamic update we have to perform a convergence check
	and potentially run additional iterations of the static algorithm
	until it converges again.

	Assuming that no further iterations of the static algorithms
	are necessary, the complexity of the update procedure
	is $\mathcal{O}(r\,|E| + \mathcal{C})$, where $\mathcal{C}$ is the
	complexity of convergence checking (see Section~\ref{sec:static}). In reality, however,
	the procedure can be expected to perform much better:
	Especially for the first few iterations, we expect the set $S$
	of vertices visited by the BFS to be much smaller than $|V|$.
	However, this implies that effective parallelization of the dynamic graph algorithm is
	more challenging than the static counterpart.
	We mitigate this problem, by aborting the BFS
	if $|S|$ becomes large and
	just update the $\omega_i$ scores unconditionally for all nodes.
	
	Finally, it is easy to see that the algorithm can be modified to update $\omega$ in-place
	instead of constructing a new $\omega'$ matrix. For this optimization,
	the algorithm needs to save the value of $\omega_i$ for all nodes of $S$
	before overwriting it, as this value is required for iteration $i+1$.
	For readability, we omit this modification in the pseudocode.


\section{Experiments}
	\label{sec:exp}

	\begin{table}[t]
		\centering
		\begin{threeparttable}
			\caption{Performance of the Katz algorithm, ranking criterion}
			\label{tbl:basic_perf}
			\begin{tabular}{rrrr@{\quad}||rrrr}
				$\epsilon$ & $r$\tnote{a} & Runtime\tnote{a} & Separation\tnote{b} &
				$\epsilon$ & $r$\tnote{a} & Runtime\tnote{a} & Separation\tnote{b} \\\hline
				$10^{-1}$ & 2.3 & 33.51 s & 96.189974 \% &
				$10^{-7}$ & 7.2 & 78.74 s & 99.994959 \% \\
				$10^{-2}$ & 3.0 & 42.81 s & 98.478250 \% &
				$10^{-8}$ & 7.9 & 83.28 s & 99.998866 \% \\
				$10^{-3}$ & 3.8 & 51.59 s & 99.264726 \% &
				$10^{-9}$ & 8.6 & 85.10 s & 99.998886 \% \\
				$10^{-4}$ & 4.8 & 65.99 s & 99.391884 \% &
				$10^{-10}$ & 9.2 & 89.03 s & 99.998889 \% \\
				$10^{-5}$ & 5.7 & 71.53 s & 99.992908 \% &
				$10^{-11}$ & 9.8 & 99.43 s & 99.998934 \% \\
				$10^{-6}$ & 6.5 & 70.59 s & 99.994861 \% &
				$10^{-12}$ & 10.4 & 96.86 s & 99.998934 \% \\
				Foster & 11.2 & 105.03 s & - &
				CG & 12.0 & 117.24 s & - \\
			\end{tabular}
			\begin{tablenotes}
				\item[a] Average over all instances. $r$ is the number of iterations.
				\item[b] Fraction of node pairs that are separated
					(and not only $\epsilon$-separated).
					Lower bound on the correctly ranked pairs.
					This is the geometric mean over all graphs.
			\end{tablenotes}
		\end{threeparttable}
	\end{table}
	
\subparagraph*{Implementation details}
	The new algorithm in this paper is hardware independent and
	as such we can implement it on different types of hardware with
	the right type of software support.
	Specifically, our dynamic Katz centrality requires a dynamic graph
	data structure. On the CPU we use NetworKit~\cite{staudt2016networkit};
	on the GPU we use Hornet\footnote{Hornet can be found at \url{https://github.com/hornet-gt},
		while NetworKit is available from \url{https://github.com/kit-parco/networkit}.
		Both projects are open source, including the implementations of our new algorithm.}.
	The Hornet data structure is architecture independent, though at time of writing
	only a GPU implementation exists.
	
	NetworKit consists of an optimized C++ network analysis library
	and bindings to access this library from Python.
	NetworKit contains parallel shared-memory implementations of
	many popular graph algorithms and can handle networks with
	billions of edges.

	The Hornet \cite{green-hornet},  an efficient extension to the cuSTINGER~\cite{green2016custinger} data structure, is a dynamic graph and matrix data structure designed for large scale networks and to support graphs with trillions of vertices. In contrast to cuSTINGER, Hornet better utilizes memory, supports memory reclamation, and can be updated almost ten times faster.

	In our experiments, we compare our new algorithm to
	Foster \etals heuristic and a conjugate gradient (CG) algorithm
	(without preconditioning)
	that solves Eq.~(\ref{eq:la}).
	The performance of CG could be possibly improved by employing
	a suitable preconditioner;
	however, we do not expect this to change our
	results qualitatively.
	Both of these algorithms were implemented
	in NetworKit and share the graph data structure with our
	new Katz implementation.
	We remark that for the static case,
	both CG and our Katz algorithm could be implemented
	on top of a CSR matrix data structure to improve the data locality
	and speed up the implementation.

\subparagraph*{Experimental setup}
	We evaluate our algorithms on a set of complex networks.
	The networks originate from diverse
	real-world applications and were taken from SNAP~\cite{leskovec2016snap}
	and KONECT~\cite{kunegis2013konect}.
	Details about the exact instances that we used can be found in
	Appendix~\ref{sec:instsel}.
	In order to be able to compare our algorithm to the CG algorithm,
	we turn the directed graphs in this test set into undirected graphs
	by ignoring edge directions. This ensures that the adjacency matrix
	is symmetric and CG is applicable.
	Our new algorithm itself would be able to handle directed graphs just fine.

	All CPU experiments ran on a machine with dual-socket Intel Xeon E5-2690
	v2 CPUs with 10 cores per
	socket\footnote{Hyperthreading was disabled for the experiments.} and 128 GiB RAM.
	Our GPU experiments are conducted on an NVIDIA P100 GPU which has
	56 Streaming Multiprocessors (SMs) and 64 Streaming Processors (SPs) per SM
	(for a total of 3584 SPs) and has 16GB of HBM2 memory.
	To effectively use the GPU, the number of active threads need to be
	roughly 8 times larger than the number of SPs.
	The Hornet framework has an API that enables such parallelization
	(with load balancing) such that the user only needs to write a few lines of code.

\subsection{Evaluation of the static Katz algorithm}
	In a first experiment, we evaluate the running time of our static Katz algorithm.
	In particular, we compare it to the running time of the linear algebra
	formulation (\ie the CG algorithm) and Foster \etals heuristic.
	We run CG until the residual is less than $10^{-15}$ to obtain a
	nearly exact Katz ranking
	(\ie up to machine precision; later in this section, we compare to CG runs with
	larger error tolerances).
	For Foster's heuristic, we use an error tolerance of $10^{-9}$,
	which also yields an almost exact ranking.
	For our own algorithm, we use the ranking convergence criterion (see Section~\ref{sec:static})
	and report running times
	and the quality of our correctness guarantees for different values of $\epsilon$.
	All algorithms in this experiment ran in single-threaded mode.

	Table~\ref{tbl:basic_perf} summarizes the results of the evaluation.
	The fourth column of Table~\ref{tbl:basic_perf} states the fraction
	of separated pairs of nodes.
	This value represents
	a lower bound on the correctness of ranking. Note
	that pairs of nodes that have the same Katz score
	will never be separated. Indeed, this seems to be the case for
	about 0.001\% of all pairs of nodes (as they are never separated, not even
	if $\epsilon$ is very low). Taking this into account, we can
	see that our algorithm already computes the correct ranking
	for $99\%$ of all pairs of nodes at $\epsilon = 10^{-3}$.
	At this $\epsilon$, our
	algorithm outperforms the other Katz algorithms considerably.
	
	\begin{figure}[t]
		\includegraphics{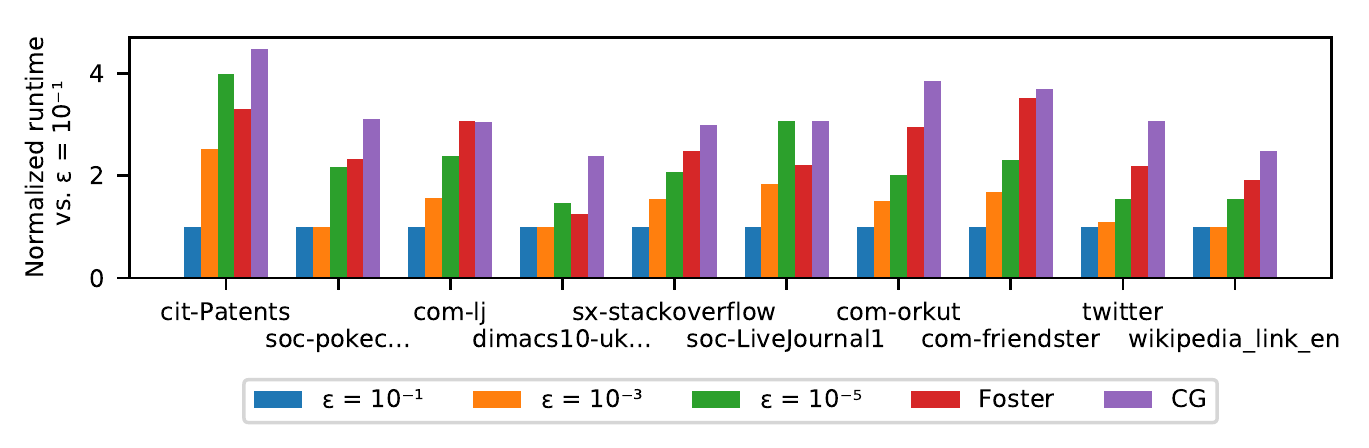}
		\caption{Katz performance on individual instances}
		\label{fig:ranktime}
	\end{figure}

	Furthermore, Table~\ref{tbl:basic_perf} shows that
	the average running time of our algorithm
	is smaller than the running time of the Foster \etal and CG algorithms.
	However, the graphs in our instance set vastly differ in
	size and originate from different applications; thus, the average
	running time alone does not give good indication for performance
	on individual graphs.
	In Figure~\ref{fig:ranktime} we report running times of
	our algorithm for the ten largest individual instances. $\epsilon = 10^{-1}$ is
	taken as baseline and the running times of all other algorithms
	are reported relative to this baseline.
	In the $\epsilon \leq 10^{-3}$ setups, our Katz
	algorithm outperforms the CG and Foster \etal algorithms on
	all instances. Foster \etals algorithm is faster than our
	algorithm for $\epsilon = 10^{-5}$ on three out of ten instances.
	On the depicted instances, CG is never faster than our algorithm,
	although it can outperform our algorithm on some small instances
	and for very low $\epsilon$.

	Finally, in Figure~\ref{fig:topk}, we present results of our Katz
	algorithm while using the top-$k$ convergence criterion. We
	report (geometric) mean speedups relative to the full ranking criterion.
	The figure also includes the approach of Nathan \etal~\cite{nathan2017guarantee}.
	Nathan \etal conducted experiments on real-world graphs and concluded
	that solving Eq.~(\ref{eq:la}) with an error tolerance of $10^{-4}$
	in practice almost always results in the correct top-$100$ ranking.
	Thus, we run CG with that error tolerance.
	However, it turns out that this approach is barely faster than
	our full ranking algorithm. In contrast to that, our top-$k$ algorithm
	yields decent speedups for $k \leq 1000$.

\subsection{Evaluation of the dynamic Katz algorithm}
	\begin{figure}[t] 
		\begin{minipage}{0.475\textwidth}
			\includegraphics{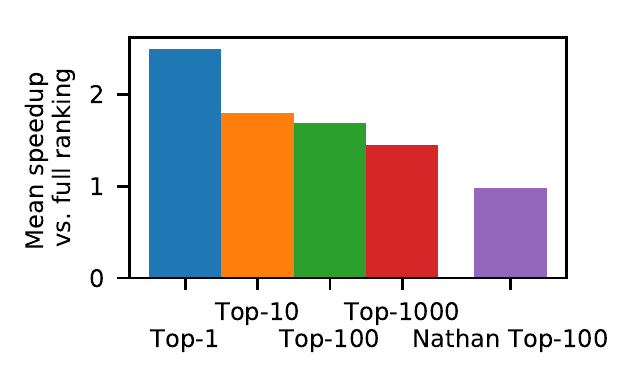}
			\caption{Top-$k$ speedup over full ranking}
			\label{fig:topk}
		\end{minipage}\hspace{0.05\textwidth}%
		\begin{minipage}{0.475\textwidth}
			\includegraphics{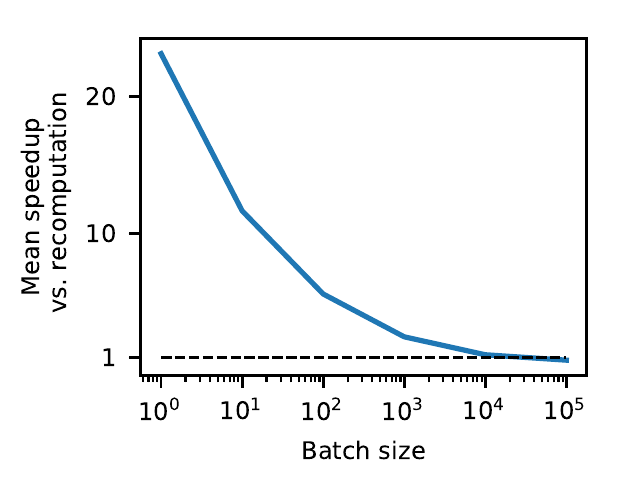}
			\caption{Dynamic update performance}
			\label{fig:dynspeedup}
		\end{minipage}
	\end{figure}
	
	In our next experiment, we evaluate the performance of
	our dynamic Katz algorithm to compute top-$1000$ rankings
	using $\epsilon = 10^{-4}$. We select $b$ random edges
	from the graph, delete them in a single batch and run our dynamic
	update algorithm on the resulting graph. We vary the batch size $b$
	from $10^0$ to $10^5$ and report the running times of the dynamic
	graph algorithm relative to recomputation.
	Similar to the previous experiment, we run the algorithms in
	single-threaded mode.
	Note that while we only show results for edge
	deletion, edge insertion is completely symmetric in Algorithm~\ref{algo:dynamic}.

	Figure~\ref{fig:dynspeedup} summarizes the results of the experiment.
	For batch sizes $b \leq 1000$, our dynamic algorithm offers a considerable
	speedup over recomputation of Katz centralities.
	As many of the graphs in our set of instances have a small diameter,
	for larger batch sizes ($b > 10000$), almost all of the vertices
	of the graph need to be visited during the dynamic update procedure.
	Hence, the dynamic update algorithm is slower than recomputation
	in these cases.

\subsection{Real-time Katz computation using parallel CPU and GPU implementations}
	Our last experiment concerns the practical running time
	and scalability of efficient parallel CPU and GPU implementations
	of our algorithm.
	For this, we compare the running times of our shared-memory
	CPU implementation with different numbers of cores.
	Furthermore, we report results of our GPU implementation.
	Because of GPU memory constraints, we could not
	process all of the graphs on the GPU. Hence, we provide the results
	of this experiment only for a subset of graphs
	that do fit into the memory of our GPU.
	The graphs in this subset have between 1.5 million and
	120 million edges.
	We use the top-$10000$ convergence criterion with $\epsilon = 10^{-6}$.

	Figure~\ref{fig:partime} depicts the results of the evaluation.
	In this figure, we consider the sequential CPU implementation as a baseline.
	We report the relative running times of the 2, 4, 8 and 16 core CPU configurations,
	as well as the GPU configuration, to this baseline.
	While the parallel CPU configurations yield moderate speedups over
	the sequential implementation, the GPU gives a significant speedup
	over the 16 core CPU
	configuration\footnote{At time of writing, our CPU implementation
		uses a sequential algorithm for partial sorting; this is
		a bottleneck in the parallel CPU configurations.}.
	Even compared to a 20 core CPU configuration (not depicted in the
	plots; see Appendix~\ref{sec:addexps}), the GPU achieves a (geometric)
	mean speedup of $10\times$.

	The CPU implementation achieves running times in the
	range of seconds; however, our GPU implementation reduces this running time
	to a fraction of a second.
	In particular, the GPU running time
	varies between 20 ms (for roadNet-PA) and 213 ms (for com-orkut),
	enabling near real-time computation of Katz centrality even for graphs with
	hundreds of millions of edges.
	\todo{Real-time argument could be improved.}

	\begin{figure}[t]
		\includegraphics{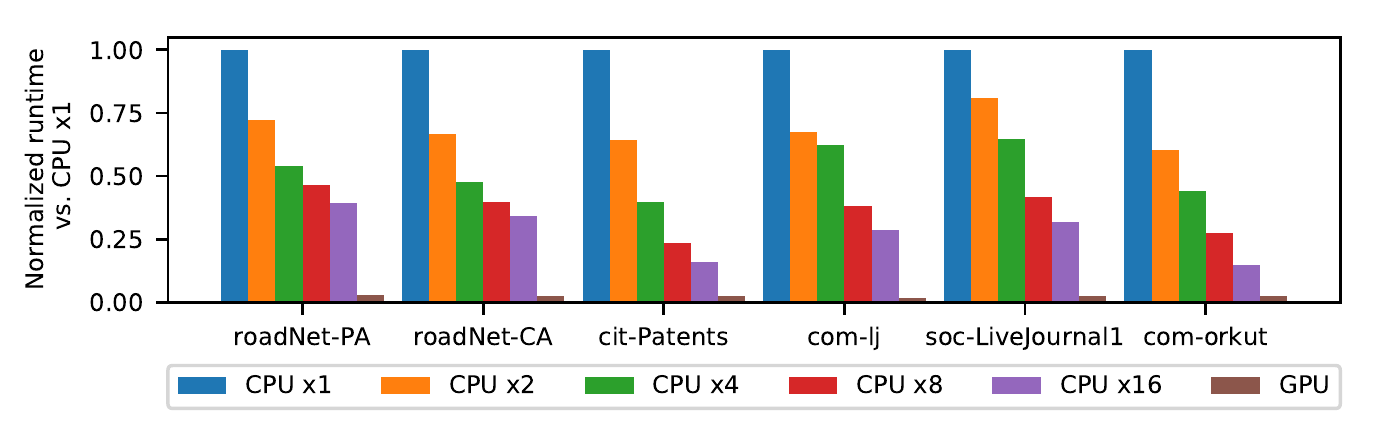}
		\caption{Scalability of parallel CPU and GPU implementations}
		\label{fig:partime}
	\end{figure}


\section{Conclusion}
	In this paper, we have presented an algorithm for Katz centrality that
	computes upper and lower bounds on the Katz score of
	individual nodes. Experiments demonstrated that our algorithm
	outperforms both linear algebra formulations and approximation
	algorithms, with speedups between 150\% and 350\%
	depending on desired correctness guarantees.

	Future work could try to provide stricter per-node bounds for Katz centrality
	to further decrease the number of iterations that the algorithm requires
	to convergence. In particular, it would be desirable to prove per-node bounds
	that do not rely on $\alpha < 1/\deg_\mathrm{max}$. On the implementation side,
	our new algorithm could be formulated in the language of GraphBLAS~\cite{kepner2016gblas}
	to enable it to run on a variety of upcoming software and hardware architectures.


\bibliographystyle{plainurl}
\bibliography{literature}


\appendix

\section{Technical proofs}
	\label{sec:proofs}
	\setcounter{theorem}{2} 

	We quickly review some technical proofs that were omitted from the main
	text of this paper.

	\begin{lemma}
		For each $x \in V$, $\ell_i(x)$ is non-decreasing in $i$ and
		$u_i(x)$ is non-increasing in $i$.
	\end{lemma}
	\begin{proof}
		The statement is trivial for $\ell_i(x)$. For $u_i(x)$ consider:
		\begin{align*}
			u_{i+1}(x) - u_{i}(x)
			&= \alpha^{i+1} \omega_{i+1}(x) + \alpha^{i+2} \omega_{i+1}(x) \gamma
					- \alpha^{i+1} \omega_i(x) \gamma \\
			&\leq \alpha^{i+1} \left(\frac1\gamma + \alpha\right) \deg_\mathrm{max} \omega_i(x) \gamma
					- \alpha^{i+1} \omega_i(x) \gamma
		\end{align*}
		Here, the inequality follows from Eq.~(\ref{eq:deg_ineq})
		in the proof of Theorem~\ref{thm:correctness}.
		To complete the proof it is sufficient to notice that
		$\left(\frac1\gamma + \alpha\right)\deg_\mathrm{max} = 1$.
	\end{proof}

	\begin{lemma}
		If $G$ is a complete graph, $u_i(x) = \katz(x)$ for all $x \in V$ and $i \in \mathbb{N}$.
	\end{lemma}
	\begin{proof}
		Consider the complete graph with $n$ vertices. Then $\omega_i(x) = (n-1)^i$
		for all $x \in V$.
		Let $\delta > 0$ be a constant so that $\alpha = \frac\delta{\deg_\mathrm{max}} = \frac\delta{n - 1}$.
		Eq.~(\ref{eq:katz}) does not converge for $\delta \geq 1$. On the other hand,
		for $\delta < 1$, the Katz centrality is given by $\katz(x) = \frac\delta{1-\delta}$.
		A short calculation (\ie rewriting the partial sum of the geometric series in $u_i(x)$)
		shows that Theorem~\ref{thm:correctness} yields the upper bound
		$u_i(x) = \katz(x)$ for all $i \in \mathbb{N}$ and $x \in V$.
	\end{proof}

\section{Instance selection}
	\label{sec:instsel}
	
	Table~\ref{tbl:instsel} states details about the instances that
	were used in our experiments in Section~\ref{sec:exp}.
	$|V|$ and $|E|$ refers to the number of vertices and edges
	in the original data (before transformation to undirected graphs).
	
	\begin{table}[h]
		\caption{Details of instances used in experimental section}
		\label{tbl:instsel}
		\centering
		\begin{tabular}{llrr}
			Name & Origin & $|V|$ & $|E|$ \\\hline
roadNet-PA & Road & 1,088,092 & 1,541,898 \\
roadNet-CA & Road & 1,965,206 & 2,766,607 \\
cit-Patents & Citation & 3,774,768 & 16,518,948 \\
soc-pokec-relationships & Social & 1,632,803 & 30,622,564 \\
com-lj & Social & 3,997,962 & 34,681,189 \\
dimacs10-uk-2002 & Road & 23,947,347 & 57,708,624 \\
sx-stackoverflow & Q\&A & 2,601,977 & 63,497,050 \\
soc-LiveJournal1 & Social & 4,847,571 & 68,993,773 \\
com-orkut & Social & 3,072,441 & 117,185,083 \\
com-friendster & Social & 65,608,366 & 437,217,424 \\
twitter & Social & 41,652,230 & 1,468,365,182 \\
wikipedia\_link\_en & Link & 13,593,032 & 1,806,067,135 \\
		\end{tabular}
	\end{table}

\section{Additional experimental data}
	\label{sec:addexps}
	For the graphs from Figure~\ref{fig:partime}, we present running
	times and exact speedups in Table~\ref{tbl:exactpar}.
	Specifically, we report absolute running times for the sequential baseline
	and speedups for the parallel and GPU implementations. Speedups
	are relative to the sequential baseline.

	\begin{table}[h]
		\caption{Absolute running times and parallel speedups for GPU instances}
		\label{tbl:exactpar}
		\centering
		\begin{tabular}{lrrrrrrr}
				& \multicolumn{7}{c}{Number of cores} \\
			Name & Sequential & 2 & 4 & 8 & 16 & 20 & GPU \\ \hline
roadNet-PA & 670 ms & 1.39$\times$ & 1.85$\times$ & 2.16$\times$ & 2.55$\times$ & 2.18$\times$ & 33.50$\times$ \\
roadNet-CA & 1018 ms & 1.50$\times$ & 2.09$\times$ & 2.51$\times$ & 2.93$\times$ & 2.69$\times$ & 40.72$\times$ \\
cit-Patents & 4751 ms & 1.55$\times$ & 2.52$\times$ & 4.23$\times$ & 6.19$\times$ & 6.41$\times$ & 41.68$\times$ \\
com-lj & 3916 ms & 1.48$\times$ & 1.60$\times$ & 2.61$\times$ & 3.51$\times$ & 4.07$\times$ & 56.75$\times$ \\
soc-LiveJournal1 & 3710 ms & 1.24$\times$ & 1.55$\times$ & 2.39$\times$ & 3.14$\times$ & 4.40$\times$ & 42.64$\times$ \\
com-orkut & 8205 ms & 1.65$\times$ & 2.26$\times$ & 3.64$\times$ & 6.69$\times$ & 6.94$\times$ & 38.52$\times$ \\
		\end{tabular}
	\end{table}

\end{document}